\documentclass[envcountsame]{llncs}
    \usepackage{amsmath}
    \usepackage{amsxtra}
    \usepackage{amstext}
    \usepackage{amssymb}
    \usepackage{latexsym}
    \usepackage{dsfont} 
\newcommand{\GF}[1]{{\mathbb F}_{#1}}

\begin{document}
\title{A direct proof of APN-ness of the Kasami functions}
\author{Claude Carlet\inst{1}, Kwang Ho Kim\inst{2,3} \and Sihem Mesnager\inst{4}}

\institute{Department of Mathematics, University of Paris VIII, F-93526 Saint-Denis, Laboratoire de G\'eom\'etrie, Analyse et Applications,  LAGA, University Sorbonne Paris Nord, CNRS, UMR 7539,  F-93430, Villetaneuse, France and  Department of informatics, University of Bergen, Norway.\\
\email{claude.carlet@gmail.com}
\and
Institute of Mathematics, State Academy of Sciences,
Pyongyang, Democratic People's Republic of Korea\\
\email{khk.cryptech@gmail.com} \and PGItech Corp., Pyongyang, Democratic People's Republic of Korea\\ \and Department of Mathematics, University of Paris VIII, F-93526 Saint-Denis, Laboratoire de G\'eom\'etrie, Analyse et Applications,  LAGA, University Sorbonne Paris Nord, CNRS, UMR 7539,  F-93430, Villetaneuse, France, and Telecom
ParisTech,  91120 Palaiseau, France.\\
\email{smesnager@univ-paris8.fr}\\} \maketitle

\begin{abstract}
Using recent results on solving the equation 
$X^{2^k+1}+X+a=0$ over a finite field $\GF{2^n}$, we address an open question raised by
the first author in WAIFI 2014 concerning the APN-ness of the Kasami functions $x\mapsto x^{2^{2k}-2^k+1}$ with $gcd(k,n)=1$, $x\in\GF{2^n}$

\noindent\textbf{Keywords:} APN function $\cdot$ Equation $\cdot$
M\"{u}ller-Cohen-Matthews (MCM) polynomial $\cdot$ Dickson
polynomial $\cdot$ Zeros of a polynomial $\cdot$ Irreducible
polynomial.
\end{abstract}

\noindent\textbf{Keywords:} 
Mathematics Subject Classification. 06E30, 11D04, 12E05, 12E12, 12E20.

\section{Introduction}

Vectorial (multi-output) Boolean functions are functions from the finite field $\mathbb{F}_{2^n}$ (of order $2^n$) to the finite field $\mathbb{F}_{2^m}$, for given positive integers
$n$ and $m$. These functions are called $(n,m)$-functions and
include the (single-output) Boolean functions (which correspond to the case $m
=1$). In symmetric cryptography, multi-output Boolean functions are called \emph{S-boxes}.
They are fundamental parts of block ciphers.
Being the only source of nonlinearity in these ciphers, S-boxes play a
central role in their robustness, by providing confusion (a requirement already mentioned by C. Shannon \cite{CC-Shannon49}), which is necessary to withstand known (and hopefully future) attacks.
 When they are used as S-boxes in block ciphers, the number $m$ of their output bits equals or is less than the number $n$ of input bits, most often. Such functions can also be used in stream ciphers, with $m$ significantly smaller than $n$, in the place of Boolean functions to speed up the ciphers.
 A survey by the first author on vectorial Boolean functions for cryptography and coding theory can be found in \cite{Cbook}. An important class of vectorial functions is that of 
\emph{almost perfect nonlinear} (APN) functions. An $(n,n)$-function $F$ is called  
APN if, for every $a\in \mathbb{F}_{2^n}^*$ and every
$b\in \mathbb{F}_{2^n}$, the equation $F(x)+F(x+a)=b$ has at most 2 solutions, that is, has $0 $ or $2$ solutions. APN functions correspond to optimal objects within other areas of mathematics (e.g. coding theory, combinatorics, and projective geometry), which makes them also interesting objects from a theoretical point of view. The first known APN functions have been power functions $F:x\mapsto x^d$, $x\in\GF{2^n}$. One class of such functions is that of Kasami APN power functions $F:x\mapsto x^{2^{2k}-2^k+1}$ with $gcd(k,n)=1$, $x\in\GF{2^n}$. The proof that Kasami functions are APN is difficult, see \cite{Janwa-Wilson1998,D98}. The first author suggested in  \cite{C14} to find a direct proof of the APN-ness of Kasami functions. This paper provides such a proof. It is structured as follows. In Section, \ref{sec:prelim} we introduce some preliminaries devoted to APN functions. Section \ref{Section4-Carlet} describes the fourth section of \cite{C14} and recalls the exact problem raised by the first author in \cite{C14}. Using the recent advances in solving the equation $X^{2^k+1}+X+a=0$ over finite fields \cite{KM19,KCM19}, we present in Section \ref{Proofs} a direct proof of the APN-ness of Kasami functions.

\section{Preliminaries and notation}\label{sec:prelim}

Let $n$ be a positive integer.  The finite field of order $2^n$ will be denoted by $\GF{2^n}$. In addition, we shall denote by $Tr$ the absolute trace function from
$\GF{2^n}$ to  $\GF{2}$ defined by
$Tr(x)=x+x^2+x^{2^2}+\cdots+x^{2^{n-1}}$.

 Differentially uniform functions are defined as follows.
\begin{definition}\label{defdu} (\cite{CC-Nyberg,CC-Nyberg-bis})
Let $n$ and $m$ be any positive integers (that we shall take in practice such that $m\leq n$) and let $\delta$ be any positive integer. An $(n,m$)-function
$F$ is called \emph{differentially $\delta$-uniform} if, for every
nonzero $a\in \mathbb{F}_{2^n}$
and every $b\in \mathbb{F}_{2^m}$, 
the equation $F(x)+F(x+a)=b$ has at most $\delta$ solutions. The minimum of such value $\delta$ for a given function $F$ is denoted by $\delta_F$ and called the \emph{differential uniformity} of $F$.
\end{definition}
The differential uniformity is necessarily even since the solutions of equation $F(x)+F(x+a)=b$ go by pairs (if $x$ is a solution of $F(x)+F(x+a)=b$ then $x+a$ is also a solution). 

When $F$ is used as an S-box inside a cryptosystem, the differential
uniformity measures its contribution to the resistance against the differential attack. The smaller is $\delta_F$, the better is the resistance.

The differential uniformity 
$\delta_F$ of any $(n,m)$-function $F$ is bounded below by $2^{n-m}$. When the differential uniformity  $\delta_F$ equals $2^{n-m}$, then $F$ is called \emph{perfect nonlinear} (PN). Perfect nonlinear functions can also be called \emph {bent functions}, since equivalently, they achieve the best possible nonlinearity $2^{n-1}-2^{\frac n2-1}$, see \cite{CC-Nyberg}.  It is well-known  that perfect nonlinear $(n,n)$-functions do not exist (precisely, they exist if and only if $n$ is even and $m\leq \frac n2$); but they do exist in other characteristics than 2 (see {\em e.g.} \cite{CC-C-D}); they are then often called \emph{planar} functions (instead of "perfect nonlinear").
\begin{definition}\label{CC-def5}(\cite{CC-BD94,CC-Nyberg-ter,CC-NK})
An $(n,n)$-function $F$ is called  
\emph{almost perfect nonlinear} (APN) if it is differentially $2$-uniform, that is, if for every $a\in {\mathbb{F}_2^n}^*$ and every
$b\in \mathbb{F}_2^n$, the equation $F(x)+F(x+a)=b$ has 0 or 2 solutions.
\end{definition}
Since $(n,m)$-functions have differential uniformity at least $2^{n-m}$ when $m\leq n/2$ ($n$ even) and strictly larger when $n$ is odd or $m>n/2$, we shall use the term of APN function only when $m=n$.  In this paper we are only dealing with APN functions. The first known APN functions have been power functions $F:x\mapsto x^d$, $x\in\GF{2^n}$. When $F$ is APN, the exponent $d$ is said to be an \emph{APN exponent}.
We present in Table \ref{Table1}, the known APN exponents up to equivalence (given any $n$, two exponents are said equivalent if they are in the same cyclotomic class of 2 modulo $2^n-1$) and up to inversion (for $n$ odd, since it is known, see e.g. \cite{Cbook}, that APN exponents are invertible modulo $2^n-1$ if and only if $n$ is odd).

\begin{table}[ht]
\begin{center} 
\caption{Known APN exponents up to equivalence (any $n$) and up to inversion ($n$ odd)}\label{Table1}
\begin{tabular}{ccc} \hline
Functions &  Exponents $d$ &  Conditions\\ \hline  
Gold   &  $2^i+1$  & $gcd(i,n)=1$ \\  
Kasami   &  $2^{2i}-2^i+1$ & $gcd(i,n)=1$ \\ 
Welch        & $ 2^t+3$  &  $n=2t+1$\\
Niho & $2^t+2^{t/2}-1$, $t$ even  & $n=2t+1$\\
 & $2^t+2^{(3t+1)/2}-1$, $t$ odd & \\
Inverse & $2^{2t}-1$& $n=2t+1$\\
Dobbertin & $2^{4t}+2^{3t}+2^{2t}+2^t-1$ & $n=5t$ \\ \hline 
\end{tabular}
\end{center} 
\end{table} 

In this paper we focus on Kasami APN functions (see \cite{Janwa-Wilson1998} and also \cite{D98}). The proof that such function is APN is difficult. The first author suggested in   \cite{C14} to find a a direct proof of APN-ness of the Kasami functions.  

\section{Description of the open question raised  by C. Carlet in \cite{C14}}\label{Section4-Carlet}
\subsection{Recall of the content of Section 4.4 in \cite{C14}}
Section 4.4 of \cite{C14} is entitled ``Find a direct proof of APN-ness of the Kasami functions in even dimension
which would use the relationship between these functions and the
Gold functions". It recalls that the proof by Hans Dobbertin in \cite{D99} of the
fact that Kasami functions $F(x) = x^{2^{2k}-2^k+1}$, where $\gcd(i,
n) = 1$, are AB (and therefore APN) for $n$ odd uses that these
functions are the (commutative) composition of a Gold function and
of the inverse of another Gold function. This proof is particularly
simple. The direct proofs in \cite{Janwa-Wilson1998} and \cite{D98} that the Kasami functions
above are APN for $n$ even are harder, as well as the determination
in \cite{DD04} (Theorem 11) of their Walsh spectrum, which also
allows to prove their APN-ness, and which uses a similar but
slightly more complex relation to the Gold functions when $n$ is not
divisible by 6. It is then written in \cite{C14} that it would be interesting to see if, for $n$ odd and
for $n$ even, these relations between the Kasami functions and the
Gold functions can lead to alternative direct proofs,
\emph{hopefully simpler}, of the APN-ness of Kasami functions.

 Since the Kasami function is a power
function, it is APN if and only if, for every $b\in \GF{2^n}$ the
system
\begin{equation}\label{eq0}
    \left\{
    \begin{array}{ll}
    X+Y&=1\\
    F(X)+F(Y)&=b
    \end{array}\right.
\end{equation}
has at most one pair $\{X, Y \}$ of solutions in $\GF{2^n}$.
\\

\noindent $\bullet$ For $n$ odd, $2^k + 1$ is coprime with
$2^n-1$ and $F(x) = G_2\circ G_1^{-1}(x)$, where $G_1(x)$ and
$G_2(x)$ are respectively the Gold functions $x^{2^k+1}$ and
$x^{2^{3k}+1}$. Hence, $F$ is APN if and only if the system
\begin{equation}\label{eq'}
    \left\{
    \begin{array}{ll}
    x^{2^k+1}+y^{2^k+1}&=1\\
    x^{2^{3k}+1}+y^{2^{3k}+1}&=b
    \end{array}\right.
\end{equation}
has at most one pair $\{x, y\}$ of solutions. Let $y = x + z$. Then
$z \neq 0$. The system \eqref{eq'} is equivalent to:
\begin{equation}
    \left\{
    \begin{array}{ll}
    \left(\frac{x}{z}\right)^{2^k}+\left(\frac{x}{z}\right)&=\frac{1}{z^{2^k+1}}+1\\
    \left(\frac{x}{z}\right)^{2^{3k}}+\left(\frac{x}{z}\right)&=\frac{b}{z^{2^{3k}+1}}+1
    \end{array}\right.
\end{equation}
or equivalently
\begin{equation}
    \left\{
    \begin{array}{ll}
    \left(\frac{x}{z}\right)^{2^k}+\left(\frac{x}{z}\right)=\frac{1}{z^{2^k+1}}+1\\
    \frac{1}{z^{2^k+1}}+1+\left(\frac{1}{z^{2^k+1}}+1\right)^{2^k}+\left(\frac{1}{z^{2^k+1}}+1\right)^{2^{2k}}=\frac{b}{z^{2^{3k}+1}}+1
    \end{array}\right.
\end{equation}
that is, by simplifying and multiplying the second equation by
$z^{2^{3k}+2^{2k}}$:
\begin{equation}
    \left\{
    \begin{array}{ll}
    \left(\frac{x}{z}\right)^{2^k}+\left(\frac{x}{z}\right)=\frac{1}{z^{2^k+1}}+1\\
   z^{2^{3k}+2^{2k}-2^k-1}+z^{2^{3k}-2^k}+1=bz^{2^{2k}-1}
    \end{array}\right.
\end{equation}
that is, denoting $v = z^{2^{2k}-1}$ and $c = b + 1$:
\begin{equation}
    \left\{
    \begin{array}{ll}
    \left(\frac{x}{z}\right)^{2^k}+\left(\frac{x}{z}\right)&=\frac{1}{v^{\frac{1}{2^k-1}}}+1\\
   (v+1)^{2^k+1}+cv&=0
    \end{array}\right.
\end{equation}
Proving that $F$ is APN is equivalent to proving that, for every
$c\in \GF{2^n}$, the second equation can be satisfied by at most one
value of $v$ such that the first equation can admit solutions, i.e.
such that $Tr \left(\frac{1}{z^{2^k+1}}+1\right)=0$. It is recalled in \cite{C14} that Reference
\cite{HK10} studies the equation $x^{2^k+1} + c (x + 1) = 0$, but observed that this does not seem to allow completing a proof. Then is stated the:\\

\noindent \textbf{Open Question 1:} For $n$ odd, is it possible to complete
this proof?\\

\noindent $\bullet$ 
For $n$ even, note first that System~\eqref{eq0} has a solution such that $X = 0$ or $Y = 0$ if
and only if b = 1. We restrict now ourselves to the case where $n$
is not divisible by 6. Then $(\frac{2^n-1}{3},3)=1$ and every
element $X$ of $\GF{2^n}^*$ can be written (in 3 different ways) in
the form $\omega x^{2^k+1}, \omega \in \GF{4}^*,  x \in \GF{2^n}^*$.
Indeed, the function $x\longmapsto  x^{2^k+1}$ is 3-to-1 from
$\GF{2^n}^*$ to the set of cubes of $\GF{2^n}^*$,  and every
integer $i$ is, by the B\'{e}zout theorem, the linear combination
over $\mathbb{Z}$ of $\frac{2^n-1}{3}$ and 3; the element $\alpha^i$
of $\GF{2^n}^*$ (where  $\alpha$ is primitive) is then the product of a
power of $\alpha^{\frac{2^n-1}{3}}$ and of a power of $\alpha^3$.
Note that $2^{2k}-2^k+1=(2^k+1)^2-3\cdot 2^k$ is divisible by 3. So, $F$ is APN if and only if the system
\begin{equation}\label{eq'_even}
    \left\{
    \begin{array}{ll}
    \omega x^{2^k+1}+\omega' y^{2^k+1}&=1\\
    x^{2^{3k}+1}+y^{2^{3k}+1}&=b
    \end{array}\right.,
\end{equation}where
$\omega, \omega' \in \GF{4}^*$ and  $x,y \in \GF{2^n}^*$, has no
solution for $b = 1$ and has at most one pair $\{\omega x^{2^k+1},
\omega' y^{2^k+1}\}$ of solutions for every $b \neq 1$. We consider
the case $x^{2^k+1}=y^{2^k+1}$ (and $\omega \neq \omega'$) apart. In
this case, the first equation $(\omega + \omega') x^{2^k+1} = 1$ is
equivalent to $x\in \GF{4}^*$ and $\omega + \omega'=1$. Then because
of the second equation, for $b = 0$, we have
two solutions such that $x^{2^k+1}=y^{2^k+1}$ (since $\omega$ and
$\omega'$ are nonzero) and for $b\neq 0$ we have none. Hence, $F$ is
APN if and only if the system
\begin{equation}\label{eq'_even'}
    \left\{
    \begin{array}{ll}
    \omega x^{2^k+1}+\omega' y^{2^k+1}&=1\\
    x^{2^{3k}+1}+y^{2^{3k}+1}&=b
    \end{array}\right.
\end{equation}
where $\omega, \omega' \in \GF{4}^*$ and  $x,y \in \GF{2^n}^*$ are
such that $x^{2^k+1}\neq y^{2^k+1}$, has no solution for $b \in
\GF{2}$ and has at most one pair $\{\omega x^{2^k+1}, \omega'
y^{2^k+1}\}$ of solutions for every $b \notin \GF{2}$. Since
$x^{2^k+1}\neq y^{2^k+1}$, we can as above denote $y = x + z$ where
$z \neq 0$, $v = z^{2^{2k}-1}$ and $c = b + 1$, and we obtain the
system:

\begin{equation}\label{eq_e_0}
    \left\{
    \begin{array}{ll}
    (\omega+\omega')\left(\frac{x}{z}\right)^{2^k+1}+\omega'\left(\frac{x}{z}\right)^{2^k}+\omega'\left(\frac{x}{z}\right)&=\frac{1}{z^{2^k+1}}+\omega'\\
   \left(\frac{x}{z}\right)^{2^{3k}}+\left(\frac{x}{z}\right)&=\frac{b}{z^{2^{3k}+1}}+1
    \end{array}\right.
\end{equation}
where $z\neq 0, v\neq 0$. 

\begin{remark} System~\eqref{eq_e_0} is slightly different from the system obtained in \cite{C14}; it is better adapted to finding a direct proof of the APN-ness of Kasami functions.\end{remark}

As in the case of $n$ odd, it is written in \cite{C14} that the results of  \cite{HK10} do not seem to allow completing a direct proof of APN-ness. Then is stated the:\\

\noindent \textbf{Open Question 2:} For $n$ even not divisible by 6, is it
possible to complete this proof?\\
\textbf{Open Question 3:} For $n$ divisible by 6, is it possible to
adapt the method?\\

\section{Proofs of APN-ness of
Kasami functions}\label{Proofs}
In this section, we complete the direct proofs of the APN-ness of
Kasami functions, for $n$ odd and for $n$ even.

 Let
$q=2^k$. We will use the following result.
\begin{lemma} (Lemma 7 of \cite{KCM19})\label{lem}
Let $(n,k)=1$. The equation $X^{q+1}+X+a=0$ has only 0, 1 or 3
solutions in  $\GF{2^n}$. If the equation $X^{q+1}+X+a=0$ has three
solutions in $\GF{2^n}$, then there exists an $u\in
\GF{2^n}\setminus\GF{2^2}$ such that
$a=\frac{(u+u^q)^{q^2+1}}{(u+u^{q^2})^{q+1}}$. Furthermore, in that
case the three solutions are $x_1=\frac{1}{1+(u+u^q)^{q-1}}$,
$x_2=\frac{u^{q^2-q}}{1+(u+u^q)^{q-1}}$ and
$x_3=\frac{(u+1)^{q^2-q}}{1+(u+u^q)^{q-1}}$.
\end{lemma}
\begin{proof}
The fact the $X^{q+1}+X+a=0$ has only 0, 1 or 3 solutions in
$\GF{2^n}$ is well known (see e.g. \cite{HK10,KM19}). The fact that
if $X^{q+1}+X+a=0$ has three solutions in $\GF{2^n}$, then there
exists an $u\in \GF{2^n}\setminus\GF{2}$ such that
$a=\frac{(u+u^q)^{q^2+1}}{(u+u^{q^2})^{q+1}}$, is a direct
consequence of Proposition 5 and Proposition 1 in \cite{HK10}.

For $a=\frac{(u+u^q)^{q^2+1}}{(u+u^{q^2})^{q+1}}$, $u\notin \GF{2}$,
the fact that $x_1=\frac{1}{1+(u+u^q)^{q-1}}$,
$x_2=\frac{u^{q^2-q}}{1+(u+u^q)^{q-1}}$ and
$x_3=\frac{(u+1)^{q^2-q}}{1+(u+u^q)^{q-1}}$ are different solutions
to $X^{q+1}+X+a=0$ can be checked by straightforward
substitution.\qed
\end{proof}
\subsection{Case when $n$ is odd}
 Since $(n,k)=1$ and $n$ is odd, it holds that
$(q-1,2^n-1)=(q+1,2^n-1)=1$. Carlet's question can be restated
as: Prove that for every $c\in \GF{2^n}$ the following system of
equations:
\begin{equation}\label{eq}
    \left\{
    \begin{array}{ll}
    Tr\left(\frac{1}{v^{\frac{1}{q-1}}}\right)=1\\
   (v+1)^{q+1}+cv&=0
    \end{array}\right.
\end{equation}
has at most one $\GF{2^n}-$solution.
\begin{proof}
If $c=0$, then the statement is right as evidently
Equation~\eqref{eq} has the unique solution 1.
Let us then assume $c\neq0$.
 By the variable substitution
$v=c^{1/q}V+1$, the second equation becomes $V^{q+1}+V+c^{-1/q}=0$.
By Lemma~\ref{lem}, we know $V^{q+1}+V+c^{-1/q}=0$  has 0, 1 or 3
$\GF{2^n}-$solutions for any  $c\in \GF{2^n}$. If this equation has
at most one solution, then Equation~\eqref{eq} also has at most one
solution.

Let us assume that this equation has 3 solutions in  $\GF{2^n}$.
Then, by Lemma~\ref{lem} there exists an $u\in
\GF{2^n}\setminus\GF{2}$ such that
$c^{-1/q}=\frac{(u+u^q)^{q^2+1}}{(u+u^{q^2})^{q+1}}$. Furthermore,
these three solutions are:  $V_1=\frac{1}{1+(u+u^q)^{q-1}}$,
$V_2=\frac{u^{q^2-q}}{1+(u+u^q)^{q-1}}$ and
$V_3=\frac{(u+1)^{q^2-q}}{1+(u+u^q)^{q-1}}$.

Thus, the three solutions to $(v+1)^{q+1}+cv=0$ are the following:
\begin{itemize}
\item
$v_1=c^{1/q}V_1+1=\frac{(u+u^{q^2})^{q+1}}{(u+u^q)^{q^2+1}}\cdot\frac{1}{1+(u+u^q)^{q-1}}
+1=\frac{(u+u^{q^2})^{q+1}}{(u+u^q)^{q^2+1}}\cdot\frac{u+u^q}{u+u^{q^2}}
+1=
\frac{(u+u^{q^2})^{q}}{(u+u^q)^{q^2}}+1=\frac{1}{(u+u^q)^{q^2-q}}$.

\item $v_2=\frac{u^{q^2}(u+u^{q^2})^{q}}{u^q(u+u^q)^{q^2}}+1=\frac{u^{q^2}(u+u^{q^2})^{q}+u^q(u+u^q)^{q^2}}{u^q(u+u^q)^{q^2}}=\frac{u^{q^3}(u+u^q)^q}{u^q(u+u^q)^{q^2}}=\frac{u^{q^3-q}}{(u+u^q)^{q^2-q}}$.

\item $v_3=\frac{(u+1)^{q^2}(u+u^{q^2})^{q}}{(u+1)^q(u+u^q)^{q^2}}+1=\frac{(u+1)^{q^2}(u+u^{q^2})^{q}+(u+1)^q(u+u^q)^{q^2}}{(u+1)^q(u+u^q)^{q^2}}=\frac{(u+1)^{q^3}(u+u^q)^q}{(u+1)^q(u+u^q)^{q^2}}=\frac{(u+1)^{q^3-q}}{(u+u^q)^{q^2-q}}$.
\end{itemize}

But, we have
$Tr\left(\frac{1}{v_1^{\frac{1}{q-1}}}\right)=Tr\left(\frac{1}{v_2^{\frac{1}{q-1}}}\right)=Tr\left(\frac{1}{v_3^{\frac{1}{q-1}}}\right)=0$
since
\begin{itemize}
\item $\frac{1}{v_1^{\frac{1}{q-1}}}=u^q+u^{q^2}$;

\item
$\frac{1}{v_2^{\frac{1}{q-1}}}=\frac{(u+u^q)^{q}}{u^{q(q+1)}}=\frac{1}{u^q}+\frac{1}{u^{q^2}}$;

 \item $\frac{1}{v_3^{\frac{1}{q-1}}}=\frac{(u+u^q)^{q}}{(u+1)^{q(q+1)}}=\frac{(u+1)^q+(u+1)^{q^2}}{(u+1)^{q(q+1)}}=\frac{1}{(u+1)^{q}}+\frac{1}{(u+1)^{q^2}}$.

\end{itemize}

Hence, Equation~\eqref{eq} has no $\GF{2^n}-$solution in this
case.\qed

\subsection{Case when $n$ is even}

\subsubsection{A simplest direct proof}: M$\ddot{u}$ller-Cohen-Matthews polynomials are defined as follows:
\begin{equation*}
  f_{k,2^k+1}(X) := \frac{{T_k(X)}^{2^k+1}}{X^{2^k}}
\end{equation*}
where $T_k(X) := \sum_{i=0}^{k-1}X^{2^i}$. The following fact is
well-known.
\begin{lemma}\cite{CM94,KM19}
  If $(n,k)=1$ and $k$ is odd, then $f_{k,2^{k}+1}$ is a permutation
    on $\GF{2^n}$.
\end{lemma}
A very concise proof of this fact is also given by Section 6 in
\cite{DD04}, by using a classical result (by Dickson in 1896) on
Dickson polynomials.

 Now, equality $F(X)+F(X+1)+1=f_{k,2^k+1}(X+X^2)$ can be
checked by direct calculation. If $n$ is even, then $k$ is odd as
$(n,k)=1$. So $F(X)+F(X+1)$ is 2-to-1 by above Lemma, i.e., the
Kasami functions are APN.\qed

\subsubsection{Continuing the discussion:}
While a very simple direct proof for $n$ even already exists as
presented above, here we will try to continue the
discussion from Section 1.
One should keep in mind the following facts:
\begin{enumerate}
\item When $n$ is divisible by 4,
$Tr(\omega)=\omega+\omega^2+\cdots+\omega^{2^{n-1}}=0$ for each
$\omega \in \GF{4}^*$.
\item When
$n$ is even not divisible by 4, $Tr(\omega)=1$ , for each $\omega
\in \GF{4}\setminus \GF{2}$ and $Tr(1)=0$.
\item Since $k$ is odd, it
holds that $\omega^{q-1}=\omega=\omega^{\frac{1}{q-1}}$,
$\omega^{q}=\omega^2$, $\omega^{q+1}=1$, $\omega^{q^2}=\omega$ for
each $\omega \in \GF{4}^*$.
\end{enumerate}

Now, let us assume $\omega = \omega'$ i.e. $\omega + \omega'=0$. The
system \eqref{eq_e_0} is equivalent to:
\begin{equation*}
    \left\{
    \begin{array}{ll}
    \left(\frac{x}{z}\right)^{q}+\left(\frac{x}{z}\right)&=\frac{1}{\omega'z^{q+1}}+1\\
    \left(\frac{x}{z}\right)^{q^3}+\left(\frac{x}{z}\right)&=\frac{b}{z^{q^3+1}}+1
    \end{array}\right.
\end{equation*}
or equivalently with $\varpi=\frac{1}{\omega'}$
\begin{equation*}
    \left\{
    \begin{array}{ll}
    \left(\frac{x}{z}\right)^{q}+\left(\frac{x}{z}\right)=\frac{\varpi}{z^{q+1}}+1\\
    \frac{\varpi}{z^{q+1}}+1+\left(\frac{\varpi}{z^{q+1}}+1\right)^{q}+\left(\frac{\varpi}{z^{q+1}}+1\right)^{q^2}=\frac{b}{z^{q^3+1}}+1.
    \end{array}\right.
\end{equation*}
 By simplifying and multiplying the second equation by
$z^{q^3+q^2}$,
\begin{equation*}
    \left\{
    \begin{array}{ll}
    \left(\frac{x}{z}\right)^{q}+\left(\frac{x}{z}\right)=\frac{\varpi}{z^{q+1}}+1\\
    z^{q^3+q^2-q-1}+\varpi z^{q^3-q}+1=b\varpi^{2}z^{q^2-1}
    \end{array}\right.
\end{equation*}
that is, denoting $v =\varpi^2 z^{q^2-1}$ and $c = b + 1$:
\begin{equation*}
    \left\{
    \begin{array}{ll}
    \left(\frac{x}{z}\right)^{q}+\left(\frac{x}{z}\right)&=\frac{1}{v^{\frac{1}{q-1}}}+1\\
   (v+1)^{q+1}+cv&=0.
    \end{array}\right.
\end{equation*}
\end{proof}

Let $\varepsilon$ be such that $1=\varepsilon+\varepsilon^2$ that
is, $\varepsilon\in \GF{4}\setminus \GF{2}$. Then, one has
$\varepsilon^q+\varepsilon=1$ and
$\varepsilon^q=\frac{1}{\varepsilon} $.

By the same arguments as in Section 2.1, when $(v+1)^{q+1}+cv=0$ has
three solutions, there exists an $u\in \GF{2^n}\setminus \GF{2^2}$
such that $\frac{1}{v^{\frac{1}{q-1}}}\in \{u+u^q,
\frac{1}{u}+\frac{1}{u^q}, \frac{1}{u+1}+\frac{1}{(u+1)^q}\}$. Let
us define $S:=\{\frac{\varpi}{u+u^q}(u+\varepsilon)^{q+1},
\frac{\varpi}{u+u^q}(u+1+\varepsilon)^{q+1}\}$.
\begin{itemize}
\item If $\frac{1}{v^{\frac{1}{q-1}}}= u+u^q$, then $z^{q+1}=\varpi
v^{\frac{1}{q-1}}=\frac{\varpi}{u+u^q}$ and $x^{q+1}\in
\{z^{q+1}(u+\varepsilon)^{q+1}, z^{q+1}(u+1+\varepsilon)^{q+1}\}=S$.

\item If $\frac{1}{v^{\frac{1}{q-1}}}= \frac{1}{u}+\frac{1}{u^q}$,
then $z^{q+1}=\varpi v^{\frac{1}{q-1}}= \frac{\varpi
u^{q+1}}{u+u^q}$ and $x^{q+1}\in
\{z^{q+1}(\frac{1}{u}+\varepsilon)^{q+1},
z^{q+1}(\frac{1}{u}+1+\varepsilon)^{q+1}\}=\{\frac{\varpi
u^{q+1}}{u+u^q}(\frac{1}{u}+\varepsilon)^{q+1}, \frac{\varpi
u^{q+1}}{u+u^q}(\frac{1}{u}+1+\varepsilon)^{q+1}\}=S$.

\item If $\frac{1}{v^{\frac{1}{q-1}}}=
\frac{1}{u+1}+\frac{1}{(u+1)^q}$, then $z^{q+1}=\varpi
v^{\frac{1}{q-1}}= \frac{\varpi (u+1)^{q+1}}{u+u^q}\}$ and
$x^{q+1}\in \{z^{q+1}(\frac{1}{u+1}+\varepsilon)^{q+1},
z^{q+1}(\frac{1}{u+1}+1+\varepsilon)^{q+1}\}=\{\frac{\varpi
(u+1)^{q+1}}{u+u^q}(\frac{1}{u+1}+\varepsilon)^{q+1}, \frac{\varpi
(u+1)^{q+1}}{u+u^q}(\frac{1}{u+1}+1+\varepsilon)^{q+1}\}=S$.
\end{itemize}
That is, $x^{q+1}\in S$ for all cases. Thus, for
$b=\frac{u+u^{q^3}}{(u+u^q)^{q^2-q+1}}$ with $u\in \GF{4}\setminus
\GF{2}$, there are two solutions
$\{\frac{(u+\varepsilon)^{q+1}}{u+u^q},
\frac{(u+1+\varepsilon)^{q+1}}{u+u^q}\}$ with $\omega=\omega'$.

It remains  to prove that for these values of $b$ there are no
solutions with $\omega \neq \omega'$. This will require more
discussion left to the reader.

\section{Conclusion}
In this paper, we have provided a direct and simpler proof  of the APN-ness of  Kasami Functions.
This solves an open question raised by the first author at the conference WAIFI 2014, which remained unanswered during six years.

\end{document}